\documentclass[11pt]{article}
\usepackage{latexsym,amssymb,amsmath,amsthm,graphicx,float}

\newcommand{\R}{\mathbb{R}}
\newcommand{\C}{\mathbb{C}}

\newcommand{\<}{\langle}
\renewcommand{\>}{\rangle}

\newcommand{\supp}{\mbox{supp}}

\newcommand{\eps}{\varepsilon}

\newcommand{\conj}{\overline}
\newcommand{\beq}{\begin{equation}}
\newcommand{\eeq}{\end{equation}}
\renewcommand{\tilde}{\widetilde}

\newcommand{\bit}{\begin{itemize}}
\newcommand{\eit}{\end{itemize}}
\newcommand{\ben}{\begin{enumerate}}
\newcommand{\een}{\end{enumerate}}
\newcommand{\bp}{\begin{pmatrix}}
\newcommand{\ep}{\end{pmatrix}}

\newtheorem{definition}{Definition}

\newtheorem{theorem}{Theorem}
\newtheorem{lemma}{Lemma}

\newtheorem{corollary}[theorem]{Corollary}
\newtheorem{proposition}[theorem]{Proposition}
\newtheorem{remark}[subsection]{Remark}

\title{The recoverability limit for superresolution via sparsity}
\author{Laurent Demanet and Nam Nguyen}
\date{December 2014}		

\begin{document}
\maketitle

\begin{abstract}
We consider the problem of robustly recovering a $k$-sparse coefficient vector from the Fourier series that it generates, restricted to the interval $[- \Omega, \Omega]$. The difficulty of this problem is linked to the superresolution factor SRF, equal to the ratio of the Rayleigh length (inverse of $\Omega$) by the spacing of the grid supporting the sparse vector. In the presence of additive deterministic noise of norm $\sigma$, we show upper and lower bounds on the minimax error rate that both scale like $(SRF)^{2k-1} \sigma$, providing a partial answer to a question posed by Donoho in 1992. The scaling arises from comparing the noise level to a restricted isometry constant at sparsity $2k$, or equivalently from comparing $2k$ to the so-called $\sigma$-spark of the Fourier system. The proof involves new bounds on the singular values of restricted Fourier matrices, obtained in part from old techniques in complex analysis.
\end{abstract}

{\bf Acknowledgments.} This work was funded by the Air Force Office of Scientific Research and the Office of Naval Research. LD also acknowledges funding from the National Science Foundation and Total S.A.

\section{Introduction}

In this paper we consider approximations in the partial Fourier system
\[
a_j(\omega) = \frac{e^{i j \tau \omega}}{\sqrt{2 \Omega}}, \qquad \omega \in [-\Omega, \Omega],
\]
where $\tau$ is the grid spacing and $\Omega$ is the band limit. We recover Fourier series when $\Omega = \frac{\pi}{\tau}$, but for smaller values of $\Omega$ the collection $a_j(\omega)$ is non-orthogonal and redundant.

We are interested in the problem of recovering the coefficients $x_{0,j}$ that enter $k$-sparse expansions of the form
\beq\label{eq:sparse-f}
f(\omega) = \sum_{j \in T} x_{0,j} a_j(\omega) + e(\omega), \qquad |T| = k,
\eeq
from the sole knowledge of $f(\omega)$ with $\omega \in [-\Omega, \Omega]$, and where $e(\omega)$ is a perturbation of size $\| e \|_2 \leq \sigma$. The notation $|T|$ refers to the cardinality of $T$. The difficulty of this problem is governed by the superresolution factor 
\[
\mbox{SRF} \triangleq \frac{\pi}{\tau \Omega},
\]
which measures the number of grid points covered by the Rayleigh length $\frac{\pi}{\Omega}$. This paper is concerned with the precise balance between SRF, the sparsity $k$, and the noise level $\sigma$, for which recovery of the index set $T$ and the coefficients $x_{0,j}$ is possible. 


\newpage

It is well-known that the sparse recovery problem (\ref{eq:sparse-f}) is one of the simplest mathematical models that embodies the difficulty of superresolution in diffraction-limited imaging, direction finding, and bandlimited signal processing. An important alternative would be to let $t_j$ receive any positive value in place of $j \tau$, but we do not deal with the ``off-grid" case in this paper.

Without loss of generality, and for the remainder of the paper, we consider the renormalized problem
\[
a_{j}(\theta) = \frac{e^{ i j \theta}}{\sqrt{2 \pi y}}, \qquad  \theta \in [- \pi y, \pi y],
\]
where $\theta = \tau \omega$ and $y = \frac{\tau \Omega}{\pi} = \frac{1}{\mbox{SRF}}$. We now recover Fourier series when $y = 1$. In the sequel we assume $y < 1/2$. 

\subsection{Minimax recovery theory}

Write $f = A x_0 + e$ as a shorthand for an expansion in the dictionary $A_{\theta, j} = a_j(\theta)$ with coefficients $x_{0,j}$, plus some noise $e$. The theory that we now present applies to general matrices\footnote{Albeit with a continuous row index. Because the column index is finite, this feature is inconsequential and does not warrant the usual complications of functional analysis.} $A$, not necessarily to partial Fourier matrices. For an index set $T$, denote by $A_T$ the restriction of $A$ to columns in $T$. Assume that the columns are unit-normed. 

The best achievable error bound on any approximation of $x_0$ from the knowledge of $f$ is linked to the concept of lower restricted isometry constant. This notion is well-known from compressed sensing, but is used here in the very different regime of arbitrarily ill-conditioned submatrices $A_T$.

\begin{definition}\label{def:RIC}(Lower restricted isometry constant)
Let $k > 0$ be an integer. Then 
\[
\varepsilon_k = \min_{T: |T|=k} \sigma_{\min}(A_T).
\]
\end{definition}

Note that $\varepsilon_k = \sqrt{1 - \delta_k}$ in the notation of \cite{candes2006robust, candes2006stable}.

Denote by $\tilde{x}$ any estimator of $x_0$ based on the knowledge of $f = A x_0 + e$. The minimax error of any such estimator, in the situation when $\| x_0 \|_0 \equiv |\supp \; x_0| = k$ and $\| e \| \leq \sigma$, is
\[
E(k, \sigma) = \inf_{\tilde{x}} \;\; \sup_{x_0 : \| x_0 \|_0 = k} \;\; \sup_{e : \| e \| = \sigma} \| \tilde{x} - x_0 \|.
\]
The minimax error is tightly linked to the value of the lower restricted isometry constant \emph{at sparsity level $2k$}. We prove the following result in Section \ref{sec:minimax}.

\begin{theorem}\label{teo:minimax}
Let $k > 0$ be an integer, and let $\sigma > 0$. We have the bounds
\[
\frac{1}{2} \frac{1}{\varepsilon_{2k}} \sigma \leq E(k, \sigma) \leq 2 \frac{1}{\varepsilon_{2k}} \sigma.
\]
\end{theorem}

An estimator $\tilde{x}$ is said to be minimax if its error $\sup_{x_0 : \| x_0 \|_0 = k} \;\; \sup_{e : \| e \| = \sigma} \| \tilde{x} - x_0 \|$ obeys the same scaling as $E$, up to a multiplicative constant.

The relevance of $\eps_{2k}$ is clear: it is the error magnification factor of any minimax estimator of $x_0$. Estimation of a general $k$-sparse coefficient sequence is possible if and only if $\sigma$ is small in comparison to $\eps_{2k}$.

\subsection{The lower restricted isometry constant}

The analysis that we present in this paper reveals that $\eps_n$ is controlled by the superresolution factor via the quantity $c(y) = \sin(\frac{\pi y}{2}) = \sin(\frac{\pi}{2 \; \mbox{\scriptsize SRF}})$.

\begin{theorem}\label{teo:eps}
There exist $C > 0$ and $y^* > 0$ such that, for all $0<y < y^*$, and with $c(y) = \sin(\pi y / 2)$,
\[
C \left( \frac{c(y)}{4} \right)^{n} \leq \eps_{n+1} \leq 4 \, c(y)^{n}.
\]
\end{theorem}

We conjecture that the restriction to small $y$ is not needed for the statement to hold.  The proof is based on two distinct results that we present in Section \ref{sec:contiguous}: 
\bit
\item Lemma \ref{teo:eps2}, which establishes that, when $y$ is small, the worst-case scenario for the least singular value is when $j = 0, 1, \ldots, k-1$ (or any $k$ consecutive integers); and 
\item Lemma \ref{teo:eps1}, which provides upper and lower bounds for the least singular value in this scenario.
\eit

This paper's main result is obtained by combining theorem \ref{teo:minimax} with theorem \ref{teo:eps} when $n+1 = 2k$.

\begin{corollary}\label{teo:main}
\[
C_{1,k} (SRF)^{2k-1} \sigma \leq E(k,\sigma) \leq C_{2,k} (SRF)^{2k-1} \sigma.
\]
\end{corollary}

The proof is clear from the fact that $c(y) \asymp (SRF)^{-1}$, and from absorbing the unknown behavior of $\eps_{2k}$ for small SRF in the pre-constants. For the same reasons as above, we conjecture that the constants $C_{1,k}$ and $C_{2,k}$ do not depend on $k$. 

Note that Corollary \ref{teo:main} is the worst-case bound. There may exist large subsets of vectors $x_0$ that exhibit further structure than $k$-sparsity, and for which the recovery rate is substantially better than $(SRF)^{2k-1} \sigma$.

\subsection{Related work}

Corollary \ref{teo:main} addresses a special case of a question originally raised by Donoho in 1992 in \cite{donoho1992SR}. In that paper, Donoho recognizes that the ``sparse clumps" signal model is the right notion to achieve superresolution. Given a vector $x$, he lets $r$ for the smallest integer such that the number of nonzero elements of $x$ is at most $r$ within any consecutive subset of cardinality $r$ times the Rayleigh length. Clearly, the set of vectors that satisfies Donoho's model at level $r$ includes the $r$-sparse vectors. If $E(r,\sigma)$ denotes the minimax error of estimating a vector at level $r$, under deterministic noise of level $\sigma$ in $L^2$, then Donoho showed that
\[
C_{1,r} (SRF)^{2r-1} \sigma \leq E(r,\sigma) \leq C_{2,r} (SRF)^{2r+1} \sigma.
\]
Corollary \ref{teo:main} is the statement that there is no gap in this sequence of inequalities --- and that Donoho's lower bound gives the correct scaling --- albeit when $r$ is understood as sparsity rather than the more general (and more relevant) ``sparse clumps" model. It would be very interesting to close the exponent gap in the latter case as well.

Around the same time, Donoho et al. \cite{donoho1992maximum} established that perfect recovery of $k$-sparse \emph{positive} vectors was possible from $2k$ low-frequency \emph{noiseless} measurements, and that the mere positivity requirement is a sufficient condition to obtain unique recovery. It is worth comparing this result to very classical work on the trigonometric moment problem \cite{grenander1958toeplitz}, where $k$ complex measurements suffice to determine $k$ real-valued phases and $k$ real-valued positive ampitudes in a model of the form (\ref{eq:sparse-f}), sampled uniformly in $\omega$. The observation that $m=2k$ is the minimum number of noiseless measurements necessary for recovery of a $k$-sparse vector is also clear from the more recent literature on sparse approximation.


The significance of $2k$ as a threshold for recovery of $k$-sparse vectors also plays a prominent role in Donoho and Elad's later work \cite{donoho2003spark}. They define the spark $s$ of a matrix $A$ to be the smallest number of linearly dependent columns, and go on to show that the representation of the form $Ax$ is unique for any $s/2$-sparse vector $x$. We explain in section \ref{sec:eps-spark} why our results can be seen as a noise-robust version of this observation: the functional inverse of the lower restricted isometry constant $\eps_k$, i.e., $k$ as a function of $\eps$, qualifies as the \emph{$\eps$-spark} $s_{\eps}$ of $A$, and equals twice the sparsity level of vectors $x$ that are robustly recoverable from $A x$.

It should be emphasized that our analysis concerns the situation when data are available for all $\omega \in [-\Omega, \Omega]$, i.e., in the continuum. The same results hold for finely sampled $\omega$, though it is not the purpose of this paper to discuss precisely what sampling condition will lead to the same scaling of the minimax error. 
For superresolution, it appears that the bandwidth parameter plays a more central role in the recovery scaling than the number of measurements.

A resurgence of interest in the superresolution problem was spurred by the work of Cand\`{e}s and Fernandez-Granda, who showed that $\ell_1$ minimization\footnote{Or its continuous counterpart, the total variation of a measure, in the gridless case.} is able to superresolve spikes that are isolated, in the sense that their distance is at least a constant times the Rayleigh length \cite{candes2012towards, candes2012noisy, fernandez2013support}. In this paper's language, their stability estimate reads $E \lesssim (SRF)^{2r} \sigma$ with $r=1$. Related important work is in \cite{azais2013SR, castro2012superresolution, duval2013SR, tang2013spectral}. The same spike separation condition is also sufficient for other types of algorithms to perform superresolution, such Fannjiang and Liao's work on MUSIC \cite{fannjiang2011MUSIC, liao2014Super}, and Moitra's work on the matrix pencil method \cite{moitra2014Super}, where the separation constant is completely sharp. 

As we put the final touches to this paper, we also learned of the work of Morgenshtern and Cand\`{e}s \cite{Morgenshtern}, which shows that the estimate $E \lesssim (SRF)^{2r} \sigma$ continues to hold in the setting of Donoho's definition of $r$, for $\ell_1$ minimization on a grid, without the spike separation condition, and as long as $x_0$ is entrywise nonnegative. It is well-known that $\ell_1$ minimization does not generally superresolve when $x_0$ has opposite signs and $A$ selects low frequencies.

As mentioned earlier, Theorem \ref{teo:eps} is based on upper and lower bounds on the smallest singular value of $A_{\theta,j}$ when $j$ spans a sequence of $k$ consecutive integers (see lemma \ref{teo:eps1} in section \ref{sec:contiguous}.) The spectral problem for this matrix was already thoroughly studied in the theory of discrete prolate sequences by Slepian in \cite{slepian1978prolateV}, who found the asymptotic rate of decay for the eigenvalues of $A^* A$, both in the limit $N \to \infty$ and SRF $\to 0$. Lemma \ref{teo:eps1} however concerns the non-asymptotic case, and could not have been proved with the same techniques\footnote{The techniques in \cite{slepian1978prolateV} could have led to a weaker form of Theorem \ref{teo:eps}, which could have sufficed to arrive at Corollary \ref{teo:main}, but would have taken us farther from the conjecture that Corollary \ref{teo:main} holds with $C_1$ and $C_2$ independent of $k$.} as in \cite{slepian1978prolateV}. Note in passing that the usual operator of time-limiting and band-limiting, giving rise to non-discrete prolate spheroidal wave functions \cite{slepian1961prolateI, landau1961prolateII, LandauWidom}, is of a very different nature from $A$. Its column index is continuous, and its singular values decay factorially rather than exponentially.

From a practical point of view, it is clear that Corollary \ref{teo:main} is mostly a negative result. For any SRF greater than 1, the conditioning of the problem grows exponentially in the sparsity level $k$.


\section{Minimax recovery and the $\eps$-spark}\label{sec:minimax}


\subsection{Robust $\ell_0$ recovery}

Consider data $f = Ax_0 + e$ with $\| e \| \leq \sigma$, and the $\ell_0$ recovery problem
\[
(P_0) \qquad \min_{x} \| x \|_0, \qquad \| f -  Ax \| \leq \sigma.
\]
Any minimizer of $(P_0)$ generates an estimator of $x_0$ that we will use to prove the upper bound in Theorem \ref{teo:minimax}. We now show the role of the lower restricted isometry constant $\eps_{2k}$ at level $2k$ for $\ell_0$ recovery of a $k$-sparse $x_0$. 


\begin{theorem}\label{teo:ell0}
Let $k > 0$ be an integer.
\bit
\item[(i)] Let $\sigma > 0$. Let $x_0 \in \R^n$ with $\| x_0 \|_0 = k$, and let $f = Ax_0 + e$ for some $\| e \| \leq \sigma$. Then any minimizer $x$ of $(P_0)$ obeys $\| x - x_0 \| \leq \frac{2}{\varepsilon_{2k}} \sigma$

\item[(ii)] There exists $x_0 \in \R^n$ with $\| x_0 \|_0 = k$ such that $f = Ax_0$ is explained by a sparser vector rather than $x_0$ with tolerance $\varepsilon_{2k}$, i.e., there exists $x_1$ for which $\| x_1 \|_0 \leq k$, $\| x_1 - x_0 \| = 1$, and $\| f - A x_1 \| \leq \varepsilon_{2k}$.
\eit
\end{theorem}

\begin{proof}
Let $k > 0$.
\bit
\item[(i)] 
Let $x$ be a minimizer of $(P_0)$, so that $\| f - A x \| \leq \sigma$. Since $\| f - Ax_0 \| \leq \sigma$ as well, it follows that $\| A (x-x_0) \| \leq 2 \sigma$.  We also have $\| x \|_0 \leq \| x_0 \|_0 \leq k$, hence $\| x - x_0 \|_0 \leq 2 k$. By definition of the lower restricted isometry constant, this implies $\| A (x - x_0) \| \geq \varepsilon_{2k} \| x - x_0 \|$. Comparing the lower and upper bounds for $\| A (x-x_0) \|$, we conclude $\| x - x_0 \| \leq \frac{2}{\varepsilon_{2k}} \sigma$.

\item[(ii)] By definition of the lower restricted isometry constant, we may pick a vector $x$ of sparsity $\| x \|_0 = 2 k$, unit-normalized as $\| x \| = 1$, and such that $\| A x \| \leq \varepsilon_{2k}$. Threshold $x$ to its $k$ largest components in absolute value; call the resulting $k$-sparse vector $x_1$. Gather the remaining $k$ components into the $k$-sparse vector $- x_0$. Then $x = x_1 - x_0$ and $\| x_1 - x_0 \| = 1$. Let $f = A x_0$, and observe that $\| f - A x_1 \| = \| A x \| \leq \varepsilon_{2k}$. 
\eit
\end{proof}

It is not known whether any polynomial-time algorithm can reach those bounds in general. 

\subsection{Minimax recovery}

In this section we prove Theorem \ref{teo:minimax}. The upper bound follows from choosing any $\ell_0$ minimizer and applying Theorem \ref{teo:ell0}. 


For the lower bound, let $\tilde{x}(f)$ be any function of $f$. Pick $x \in \R^n$ such that $\| x \| = 1$, $\| A x \| \leq \varepsilon_{2k}$, and $\| x \|_0 = 2k$. As in the argument in the previous section, partition $x$ into two components $x_0$ and $-x_1$ of sparsity $k$, but normalize them so that $x = \frac{\eps_{2s}}{\sigma} (x_0 - x_1)$. Then we have $\|  A (x_0 - x_1) \| \leq \sigma$.

Now let $f = Ax_0$, and compute
\begin{align*}
\frac{\sigma}{\eps_{2k}} = \| x_0 - x_1 \| &= \| \tilde{x}(f) - x_0 - (\tilde{x}(f) - x_1) \| \\
&\leq \| \tilde{x}(f) - x_0 \| + \| \tilde{x}(f) - x_1 \| \\
&\leq 2 \max \{  \| \tilde{x}(f) - x_0 \|,  \| \tilde{x}(f) - x_1 \| \} 
\end{align*}
The data $f$ can be seen as derived from $x_0$, since $f = Ax_0$, but also from $x_1$, since $f = Ax_1 + e$ for some vector $e$ with $\| e \| \leq \sigma$. Hence
\begin{align*}
\frac{1}{2} \frac{\sigma}{\eps_{2k}} \leq \max \{  \| \tilde{x}(f) - x_0 \|,  \| \tilde{x}(f) - x_1 \| \} &\leq \max_{j = 0,1} \sup_{\| e_j \| \leq \sigma} \| \tilde{x}(Ax_j + e_j) - x_j \| \\
&\leq \sup_{\| x \|_0 = k} \sup_{\| e \| \leq \sigma} \| \tilde{x}(Ax + e) - x \|
\end{align*}
The lower bound $\frac{1}{2} \frac{\sigma}{\eps_{2k}}$ holds uniformly over the choice of $\tilde{x}$, which establishes the claim.

\subsection{Recovery from the $\eps$-spark}\label{sec:eps-spark}

We introduce the notion of $\eps$-spark of $A$, as a natural modification of the notion of spark introduced in \cite{donoho2003spark},  and link it to the notion of lower restricted isometry constant.

\begin{definition}\label{def:eps-spark}($\varepsilon$-spark)
Fix $\varepsilon > 0$. Then $s_{\varepsilon}$ is the largest integer such that, for every $T$, $|T| \leq s_\varepsilon$,
\[
\varepsilon \leq \sigma_{\min}(A_T). \qquad 
\]
\end{definition}

When the lower restricted isometry constant $\eps_s$ is strictly decreasing, it is easy to see that $s_{\eps_s} = s$ , i.e., $s_{\eps}$ is a composition inverse of $\eps_s$. However, we cannot in general expect better than $\eps_{s _\eps} \geq \eps$. When $\varepsilon = 0$, we recover the spark introduced in \cite{donoho2003spark}, though our $0$-spark is in fact Donoho and Elad's spark minus one\footnote{That seems to be the price to pay to get $s_{\eps_s} = s$.}. 

In other words, the definition of $\eps$-spark parallels that of spark, but replaces the notion of rank deficiency by that of being $\eps$-close, in spectral norm.

Theorems \ref{teo:minimax} and \ref{teo:ell0} can be seen as the robust version of the basic recovery result in \cite{donoho2003spark}. The following theorem is a literal transcription of Theorem \ref{teo:ell0} in the language of the $\eps$-spark. We respectively let $\lfloor a \rfloor$ and $\lceil a \rceil$ for $a$'s largest previous and smallest following integers.

\begin{theorem}
Let $\sigma > 0$.
\bit
\item[(i)] Assume that $\| x_0 \|_0 \leq \lfloor \frac{s_\delta}{2} \rfloor$ for some $\delta > 0$. Then any minimizer $x$ of   $(P_0)$ obeys $\| x - x_0 \| \leq 2 \sigma / \delta$.
\item[(ii)] Assume that $\sigma \geq \sigma_{\min}(A)$. There exists $x_0$ such that $\| x_0 \|_0 \geq \lceil \frac{s_\sigma}{2} \rceil$, for which $f = Ax_0$ can be approximated by a sparser vector than $x_0$, in the sense that there exists $x$ such that $\| x \|_0 \leq \| x_0 \|_0$, $\| x - x_0 \| = 1$, and $\| f - A x \| \leq \sigma$.
\eit
\end{theorem}

In other words, the sharp recovery condition comparing the noise level with the lower restricted isometry constant at level $2k$, namely
\[
\sigma \sim \eps_{2 \| x_0 \|_0},
\]
can be rephrased as the comparison of the sparsity level to half the $\sigma$-spark, as
\[
\| x_0 \|_0 \sim \frac{s_\sigma}{2}.
\]
These two points of view are equivalent.

\section{Consecutive atoms}\label{sec:contiguous}

In this section we prove Theorem \ref{teo:eps}. We return to the case $A_{\theta,j} = a_j(\theta) = \frac{e^{ i j \theta}}{\sqrt{2 \pi y}}$.

Any upper bound on $\sigma_{\min}(A_T)$ provides an upper bound on $\eps_{n+1}$ when $|T| = n+1$. However, in order to get a lower bound on $\eps_{n+1}$, we need to control $\sigma_{\min}(A_T)$ for every $T$ of cardinality $n+1$. The following lemma establishes that $T = \{ 0, 1, \ldots, n \}$ gives rise to the lowest $\sigma_{\min}(A_T)$, at least in the limit $y \to 0$. The proof is postponed to Section \ref{sec:non-contiguous}.

\begin{lemma}\label{teo:eps2}
There exists $y^{*} > 0$ such that, for all $0 <y < y^*$, the minimum of $\sigma_{\min}(A_T)$ over $T : |T| = n+1$ is attained when $T = \{ 0, 1, \ldots, n \}$.
\end{lemma}

It therefore suffices to find lower and upper bounds on the least singular value of $A_{\theta,j}$, as a semi-continuous matrix with row coordinate $\theta \in [-\pi y, \pi y]$ and column index $0 \leq j \leq n$. The result that we prove in this section is as follows.

\begin{lemma}\label{teo:eps1}
Let $T = \{ 0, 1, \ldots, n \}$ and $c(y) = \sin(\pi y / 2)$. There exists $C > 0$ such that
\[
C \left( \frac{c(y)}{4} \right)^{n} \leq \sigma_{\min}(A_T) \leq 4 \, c(y)^{n}.
\]
\end{lemma}

The singular values of $A_T$ are the square roots of the eigenvalues of the section $0 \leq j_1, j_2 \leq n$ of the Gram matrix
\[
G_{j_1 j_2} = \int_{-\pi y}^{\pi y} a_{j_1}(\theta) \conj{a}_{j_2}(\theta) \, d\theta.
\]
A detour through complex analysis will provide tools that will help understand the eigenvalues of $G$.

\subsection{Preliminaries on complex analysis and Szeg\H{o}'s theory}
\label{subsec:preliminaries}

In the sequel we rely on the characterization of $G$ as a Toeplitz form for the Lebesgue measure on a circle arc in the complex plane. 
Notice that $a_j(\theta) = \frac{1}{\sqrt{2 \pi y}} \, z^j$ with $z = e^{i \theta}$. Let $\Gamma$ be the circle arc 
\[
\Gamma = \{ z : |z|=1, \; - \pi y \leq \arg z \leq \pi y \}.
\]
Its length is $L = 2 \pi y$. Consider the arclength inner product
\beq\label{eq:IP-Gamma}
\< f, g \> = \frac{1}{L} \int_{\Gamma} f(z) \overline{g(z)} |dz|,
\eeq
and the corresponding norm $\| f \| = \sqrt{\< f,f \>}$. On the unit circle, $|dz| = \frac{1}{iz} dz$. With this inner product, we can understand $G$ as the Gram matrix of the monomials:
\[
G_{j_1, j_2} = \< z^{j_1}, z^{j_2} \>.
\]


The orthogonal (Szeg\H{o}) polynomials $\{ p_n(z) \}$ on $\Gamma$ play an important role. They are defined from applying the Gram-Schmidt orthogonalization on the monomials $\{ z^n \}$, resulting in
\[
\< p_m, p_n \> = \delta_{mn}.
\]
Denote by $k_n$ the coefficient of the highest power of $p_n(z)$, i.e., $p_n(z) = k_n z^n + \ldots$. Observe that the $p_n(z)$ are extremal in the following sense.

\begin{lemma}\label{teo:extremal} (Christoffel variational principle) Let $M_n \subset P_n$ be the set of degree-$n$ monic\footnote{With coefficient of the leading power equal to one.} polynomials over $\Gamma$. Then
\beq\label{eq:extremal}
\min_{\pi \in M_n} \| \pi \| = \frac{1}{k_n}
\eeq
The unique minimizer is $k_n^{-1} p_n(z)$.
\end{lemma}

\begin{proof}
Let $\pi(z) = \sum_{m = 0}^n \lambda_m p_m(z)$ with $\lambda_n k_n = 1$. By orthonormality,
\[
\| \pi \|^2 = \sum_{m=0}^n \lambda_m^2.
\]
Under the constraint $\lambda_n = 1/k_n$, this quantity is minimized when $\lambda_0 = \ldots = \lambda_{n-1}$. In that case the minimizer is $\lambda_n p_n(z) = k_n^{-1} p_n(z)$ and the minimum is $\lambda_n^2 = k_n^{-2}$.
\end{proof}

In order to quantify $k_n$, we need to better understand the asymptotic properties of $p_n(z)$ at infinity. Consider the analytic function $z = \phi(w)$ which maps $|w| > 1$ conformally onto to the exterior of $\Gamma$, such that $w = \infty$ is preserved, and such that the orientation at $\infty$ is preserved. It has the explicit expression
\beq\label{eq:phiw}
\phi(w) = w \, \frac{cw + 1}{w + c},
\eeq
with 
\beq\label{eq:c}
c = \sin \frac{\pi y}{2}.
\eeq
Indeed, it can be seen that
\[
\phi(e^{i \theta}) = \exp \left( 2 i \arg \left( e^{i \theta} + \frac{1}{c} \right) \right),
\]
with an argument that covers $[-\pi y, \pi y ]$ twice. This expression for $\phi$ is not new, see for example \cite{coleman1995faber}.


The number $c$ in (\ref{eq:c}) is the so-called capacity of $\Gamma$.

\begin{definition}\label{def:c} The capacity (or transfinite diameter) of $\Gamma$ is the coefficient of $w$ in the Laurent expansion of $\phi(w)$ at infinity.
\end{definition}

In our case, the Laurent expansion at $\infty$ is
\begin{align*}
\phi(w) &= cw + (1-c^2) + \frac{c(1-c^2)}{w+c} \\
&= cw + (1-c^2) + \sum_{n > 0} \gamma_n w^{-n},
\end{align*}
for some coefficients $\gamma_n$, hence the capacity of $\Gamma$ is indeed $c = \sin \frac{\pi y}{2}$.

A major finding in Szeg\H{o}'s theory \cite{szego1975orthpoly, grenander1958toeplitz} is the asymptotic match $p_n(z) \sim g_n(z)$ at $z = \infty$, where
\beq\label{eq:gnz}
g_n(z) = \left( \frac{L}{2 \pi} \right)^{1/2} (\Phi'(z))^{1/2} (\Phi(z))^n,
\eeq
and where $w = \Phi(z)$ the composition inverse of (\ref{eq:phiw}). In our case, we compute
\begin{align}\label{eq:Laurent-Phiz}
\Phi(z) &= \frac{z-1}{2c} + \left( \frac{(z-1)^2}{4c^2} + z \right)^{1/2}  \notag \\
&= \frac{z}{c} + \frac{c^2-1}{c} + \sum_{n > 0} \delta_n z^{-n}.
\end{align}
The extremities of $\Gamma$ are branch points for the square root, and the branch cut should be on $\Gamma$ itself for $\Phi(z)$ to be analytic outside $\Gamma$. 

Recall that $L = 2 \pi y$. Matching asymptotics at infinity yields
\[
p_n(z) \sim k_n z^n, \qquad \sqrt{y} \, (\Phi'(z))^{1/2} (\Phi(z))^n \sim \sqrt{y} \, c^{-n-\frac{1}{2}} z^n,
\]
hence we anticipate that
\beq\label{eq:kn-asymp}
k_n \sim \sqrt{y} \, c^{-n-\frac{1}{2}}
\eeq
as $n \to \infty$. We formulate a non-asymptotic version of this result in the next subsection.

An important proof technique in the sequel is the Szeg\H{o} kernel $K(z,z_0)$. The Hardy space $H^2(\Omega)$, where $\Omega$ extends to $\infty$ and has boundary $\Gamma$, is the space of analytic functions, bounded at infinity, and with square-integrable trace on $\Gamma$. 

\begin{definition}\label{def:szego-kernel}
The Szeg\H{o} kernel $K(z,\zeta)$ relative to the exterior $\Omega$ of a Jordan curve or arc $\Gamma$ is the reproducing kernel for $H^2(\Omega)$, i.e., the unique function $K(z,\cdot) \in H^2(\Omega)$ such that, for all $F \in H^2(\Omega)$,
\beq\label{eq:reprod}
F(z) = \frac{1}{L} \int_{\Gamma} F(\zeta) \overline{K(\zeta, z)} |d\zeta|, \qquad z \in \Omega.
\eeq
\end{definition}

We would have liked to have found the following result in the literature.

\begin{proposition}\label{teo:Kzzeta} Let $\Gamma$ be the image of the unit circle $|w| = 1$ under the conformal map $z = \phi(w)$, and assume that $\Gamma$ is a Jordan arc.
Assume that $\phi$ is one-to-one and invertible for $z$ outside $\Gamma$, and let $w = \Phi(z)$. Then the Szeg\H{o} kernel obeys
\beq\label{eq:K}
K(\zeta, z) = \frac{L}{\pi} \left( \Phi'(\zeta)  \overline{\Phi'(z)}\right)^{1/2} \frac{\Phi(\zeta) \overline{\Phi(z)}}{\Phi(\zeta) \overline{\Phi(z)} - 1}.
\eeq



\end{proposition}
\begin{proof}
The transformation law for the Szeg\H{o} kernel under a conformal map $w = \Phi(z)$ is \cite{KerzmanStein}
\beq\label{eq:conformalK}
K(z',z) = \left( \Phi'(z') \right)^{1/2} K_0(\Phi(z'), \Phi(z)) \left( \overline{\Phi'(z)} \right)^{1/2},
\eeq
where $K_0$ is for the pre-image $\Gamma_0$ (a Jordan curve) of $\Gamma$. The formula assumes that $K$ and $K_0$ are reproducing for the arclength inner products without prefactors, both in $w$ and $z$. In our setting, the desired $K$ is however normalized for (\ref{eq:reprod}), which involves a $1/L$ prefactor, and a single rather than double traversal of the arc $\Gamma$. Our desired $K$ is therefore $2 L$ times the right-hand-side in (\ref{eq:conformalK}).

In our case $\Gamma_0$ is the unit circle. It suffices therefore to show that the Szeg\H{o} kernel for the exterior of the unit circle is
\beq\label{eq:K_0}
K_0(w',w) = \frac{1}{2 \pi} \frac{w' \overline{w}}{w' \overline{w} - 1}.
\eeq
Recall Cauchy's integral formula for bounded analytic functions in the exterior of the unit circle:
\[
\frac{1}{2 \pi i} \oint \frac{f(w')}{w' - w} dw' = \left\{ \begin{array}{ll}
         f(\infty) - f(w) & \mbox{if $|w| > 1$};\\
        f(\infty) & \mbox{if $|w| < 1$}.\end{array} \right.
\]
Note that $dw' = i w' |dw'|$ and $w' = 1 / \overline{w'}$ on the unit circle. Evaluate the Cauchy formula at $w = 0$ in order to obtain $f(\infty)$, then simplify the formula for the case $|w| > 1$ as
\begin{align*}
f(w) &= \frac{1}{2 \pi} \oint \left( 1 - \frac{1}{1 - w \overline{w'}} \right) f(w') |dw'| \\
&= \frac{1}{2 \pi} \oint  \frac{w \overline{w'}}{w \overline{w'} - 1} f(w') |dw'|.
\end{align*}
This expression is of the form 
\[
f(w) = \oint f(w') \overline{K_0(w',w)} |dw'|,
\]
with $K_0$ given by (\ref{eq:K_0}). To complete the proof, we must observe that $K_0(w,w')$ is analytic and bounded in $w'$, hence a member of $w'$ in $H^2(\Omega)$ as required by definiton \ref{def:szego-kernel}. (This point is important: the Cauchy kernel doubles as Szeg\H{o} kernel only for the unit circle.)

\end{proof}

The limits as $\zeta$ and/or $z \to \infty$ exist and are finite since $K$ is an element of $H^2(\Omega)$. We also note that the kernel $K$ is extremal in the following sense. 

\begin{lemma}\label{teo:FK} (Widom \cite{widom1969extremalpoly})
Consider
\[
\mu = \inf_{F} \int_{\Gamma} |F(z)|^2 |dz|,
\]
where the infimum is over $F \in H^2(\Omega)$ such that $F(z_0) = 1$ for some $z_0 \in \overline{\C}$ (the extended complex plane including $z=\infty$). The infimum is a minimum, the extremal function is unique, and obeys
\[
F(z) = \frac{K(z,z_0)}{K(z_0,z_0)}.
\]
\end{lemma}


\subsection{Non-asymptotic bounds on the coefficient $k_n$}

\begin{theorem}\label{teo:kn-main}
With $c = \sin \frac{\pi y}{2}$, we have
\[
\frac{c}{2y} c^{2n} \leq k^{-2}_n \leq 4(1+2y)^2 c^{2n}.
\]
\end{theorem}

\begin{proof}
The proof of the lower bound is essentially an argument due to Widom \cite{widom1969extremalpoly} that we reproduce for convenience. Let $\pi(z) = k_n^{-1} p_n(z)$. From the characterization of $\Phi(z)$ in (\ref{eq:Laurent-Phiz}), and since $\pi$ is monic, we get
\[
\lim_{z \to \infty} (c \Phi(z))^{-n} \pi(z) = 1.
\]
Consider now the quantity
\[
J = \int_{\Gamma} | (c \Phi(z))^{-n} \pi(z) |^2 \, |dz|.
\]
Since $|\Phi(z)| = 1$ on $\Gamma$, and using lemma \ref{teo:extremal}, we obtain
\[
J = c^{-2n} \int_{\Gamma} | \pi(z) |^2 \, |dz| = c^{-2n} L k_n^{-2}.
\]
On the other hand, we can write the lower bound
\[
J \geq \inf_F \int_{\Gamma} |F(z)|^2 \, |dz| \equiv \mu,
\]
where the infimum is over all $F$ in the Hardy space of analytic functions in the exterior of $\Gamma$, square integrable over $\Gamma$; and such that $F(\infty) = 1$. We can invoke Lemma \ref{teo:FK} and Proposition \ref{teo:Kzzeta} to obtain the unique extremal function
\[
F(z) = \frac{K(z,\infty)}{K(\infty, \infty)} =  (c \Phi'(z))^{1/2}.
\]
We can compute the value of $\mu$ by hand:
\begin{align*}
\mu = \int_{\Gamma} |F(z)|^2 \, |dz| &= \frac{1}{2} \oint_{\Gamma} |F(z)|^2 \, |dz| \\
&= \frac{c}{2} \oint_{\Gamma} | \Phi'(z) | \, |dz| \\
&= \frac{c}{2} \oint_{|w| = 1} |dw| = c \pi.
\end{align*}
The factor 1/2 in the first line owes to the fact that, in order to change variables from $z$ to $w$, the curve $\Gamma$ is traversed \emph{twice} as $w$ traverses the unit circle. We can now combine the various bounds to obtain 
\[
k_n^{-2} \geq \frac{c \pi}{L} c^{2n} = \frac{c}{2y} c^{2n}.
\]


The proof of the upper bound is somewhat trickier and does not follow the standard asymptotic arguments of Szeg\H{o} \cite{szego1975orthpoly, grenander1958toeplitz} and Widom \cite{widom1969extremalpoly}.
We use Lemma \ref{teo:extremal}, and invoke the classical fact that the so-called Faber polynomial $\Phi_n(z)$ is an adequate substitute for $p_n(z)$, with comparable oscillation and size properties. In this context, we define $\Phi_n(z)$ as the polynomial part of the Laurent expansion at infinity of the function $(\Phi(z))^n$. From (\ref{eq:Laurent-Phiz}), we observe that
\[
\Phi_n(z) = c^{-n} z^n + \mbox{ lower-order terms}.
\]
The monic version of $\Phi_n(z)$, for use in place of the minimizer in (\ref{eq:extremal}), is $f_n(z) = c^{n} \Phi_n(z)$. We now make use of a relatively recent inequality due to Ellacott \cite{Ellacott}, (which in turn owes much to a characterization of $\Phi_n(z)$ due to Pommerenke \cite{Pommerenke}), 
\[
\max_{z \in \Gamma} | \Phi_n(z) | \leq \frac{V}{\pi},
\]
where $V$ is the so-called total rotation of $\Gamma$, defined as the total change in angle as one traverses the curve, with positive increments regardless of whether the rotation occurs clockwise or counter-clockwise. In other words, if $\theta(z)$ is the angle that the tangent to $\Gamma$ at $z$ makes with the horizontal, then $V$ is the total variation of $\theta(z)$, or
\[
V = \int_{\Gamma} | d \theta(z) |.
\]
In the case of a circle arc of opening angle $2 \pi y$, it is easy to see that $V = 2\pi (1 + 2 y)$. 

We conclude with the sequence of bounds
\begin{align*}
k_n^{-2} &\leq \| f_n \|^2 \\
&= c^{2n} \frac{1}{L} \int_{\Gamma} |\Phi_n(z)|^2 |dz| \\
&\leq c^{2n} \frac{V^2}{\pi^2} \\
&\leq 4 (1+2y)^2 c^{2n}.
\end{align*}

\begin{remark}
The exact asymptotic rate for $k_n^{-2}$ as $n \to \infty$ can be inferred from the work of Widom \cite{widom1969extremalpoly} in the same fashion as above; it is
\[
k_n^{-2} \sim \frac{c}{y} c^{2n}.
\]
However, favorable inequalities for small $n$ are not readily available from those arguments. The reason for the factor 2 discrepancy between the lower bound in Theorem \ref{teo:kn-main} and the asymptotic rate can be traced to the fact that $\Gamma$ is a Jordan arc (with empty interior), not a Jordan curve. It is for the same reason that the asymptotic expression (\ref{eq:gnz}) differs from that given by Szeg\H{o} in \cite{szego1975orthpoly}, p. 372, by a factor $1/\sqrt{2}$.

\end{remark}

\end{proof}

\subsection{Upper bound on the smallest singular value}

Let
\[
A_n = \mbox{span} \{ a_{j}(\theta) ; 0 \leq j \leq n \},
\]
and $P_n$ be the orthoprojector on $A_n$. Subspace angles allow to formulate upper bounds on eigenvalues of $G$. Specifically, recall that $\| a_n \| = 1$, and consider
\[
\sin \angle (a_{n}, A_{n-1}) = d(a_{n}, A_{n-1}) = \| a_n - P_{n-1} a_n \|.
\]
The norms are in $L^2(-\pi y, \pi y)$.

\begin{lemma}\label{teo:smin-angle}
Consider a matrix $[ A \; b ]$ with columns normalized to unit norm. Then its smallest singular value obeys
\[
s_{\min}([ A \; b ]) \leq  \sin \angle (b, \mbox{Ran} \, A).
\]
\end{lemma}
\begin{proof}
Denote by $P_A$ the orthoprojector onto $\mbox{Ran} \, A$. In the matrix spectral norm $\| \cdot \|_2$,
\begin{align*}
s_{\min}([ A \; b ]) &\leq \| [ A \; b ] - [ A \; P_A b ] \|_2 \qquad \mbox{(SVD gives the best rank $(n-1)$ approximation)} \\
&= \| b - P_A b \| \\
&= \sin \angle (b, \mbox{Ran} \, A)  \qquad \mbox{(because $\| b \| = 1$)}
\end{align*}
\end{proof}

The change of variables $z = e^{i \theta}$ reveals that
\[
\| a_n - P_{n-1} a_n \| = \| z^n - P_{n-1} z^n \| = d(z^n, P_{n-1}) ,
\]
where $P_n$ is overloaded to mean the orthoprojector onto span $\{ 1, z, \ldots, z^{n} \}$; where the first norm is in $L^2(- \pi y, \pi y)$; and where the second norm is given by equation (\ref{eq:IP-Gamma}). 

It is then well-known that $d(z^n, P_{n-1})$ is accessible from the coefficient $k_n$ introduced earlier in Section \ref{subsec:preliminaries}.

\begin{lemma}\label{teo:1overkn}
\[
d(z^n, P_{n-1}) = \frac{1}{k_n},
\]
where $p_n(z) = k_n z^n + $ lower-order terms is the orthogonal polynomial introduced in Section \ref{subsec:preliminaries}. 
\end{lemma}
\begin{proof}
The Gram-Schmidt orthogonalization procedure yields
\[
p_n(z) = \frac{z^n - P_{n-1} z^n}{\| z^n - P_{n-1} z^n \|},
\]
which takes the form
\[
p_n(z) = k_n z^n + q_{n-1}(z), \qquad q_{n-1} \in P_{n-1},
\]
with $1/k_n = \| z^n - P_{n-1} z^n \| = d(z^n, P_{n-1})$. 
\end{proof}

We can now combine Lemmas \ref{teo:smin-angle} and \ref{teo:1overkn} with Theorem \ref{teo:kn-main}, and $y < 1/2$, to conclude that the least singular value of $a_j(\theta)$, with $\theta \in [-\pi y, \pi y]$ and $0 \leq j \leq n$, is upper-bounded by
\beq\label{eq:upper-bound-least-sv}
k_n^{-1} \leq 4 c^n, \qquad c = \sin(\pi y / 2).
\eeq

%
%
%

\subsection{Lower bound on the smallest singular value}

Recall that $\sigma_{\min}(A_T) = \sqrt{\lambda_{\min} (G)}$, where $G$ is the Gram matrix
\[
G_{j_1 j_2} = \< z^{j_1}, z^{j_2} \> = \frac{1}{L} \int_{\Gamma} z^{j_1} \overline{z}^{j_2} | dz |, \qquad 0 \leq j_1, j_2 \leq n.
\]

We make use of the following characterization of the eigenvalues of $G$, which according to Berg and Szwarc \cite{Berg}, was first discovered by Aitken \cite{Collar}. It was also used in the work of Szeg\H{o} \cite{szego1975orthpoly}, and that of Widom and Wilf \cite{widom1966eigHankel}.

\begin{lemma} (Aitken)
\[
\frac{1}{\lambda_{\min}(G)} = \max_P \frac{1}{2 \pi} \int_0^{2 \pi} | P(e^{i \theta}) |^2 d\theta,
\]
where the maximum is over all degree-$n$ polynomials $P(z)$ such that $\| P \| = 1$.
\end{lemma}
\begin{proof}
The variational characterization of $\lambda_{\min}(G)$ gives
\[
\lambda_{\min}(G) = \min_{\mathbf{c}} \frac{\mathbf{c}^* G \mathbf{c}} {\mathbf{c}^* \mathbf{c}} = \min_{P} \frac{\| P \|^2}{\mathbf{c}^* \mathbf{c}},
\]
where $\mathbf{c} = (c_0, \ldots, c_n)^T$, and the last min is over $P$ of the form $P(z) = \sum_{k = 0}^n c_k z^k$. For such a $P$, we can apply orthogonality of $z^n$ on the unit circle to obtain
\[
\frac{1}{2 \pi} \int_0^{2 \pi} | P(e^{i \theta}) |^2 d\theta = \mathbf{c}^* \mathbf{c}.
\]
\end{proof}


A useful bound for the growth of any such $P(z)$ away from $\Gamma$ can be obtained from the fact that the Szeg\H{o} kernel reproduces bounded analytic functions outside of $\Gamma$. The following argument was used in \cite{widom1969extremalpoly}.

\begin{lemma}
Let $P(z)$ be a polynomial of degree $n$ such that $\| P \|^2 = \frac{1}{L} \int_\Gamma |P(z)|^2 |dz| = 1$. Then
\[
|P(z)| \leq K(z,z)^{1/2} |\Phi(z)|^n.
\]
\end{lemma}
\begin{proof}
Let $F(z) = \Phi(z)^{-n} P(z)$. This analytic function obeys $\| F \|^2 = 1$, and is bounded at $z = \infty$, hence belongs to the Hardy space $H^2(\Omega)$. By Definition \ref{def:szego-kernel},
\[
F(z) = \frac{1}{L} \int_{\Gamma} \overline{K(z, \zeta)} \, F(\zeta)  |d\zeta|.
\]
By Cauchy-Schwarz, we get
\[
| F(z) | \leq \left( \frac{1}{L} \int_{\Gamma} | K(z,\zeta) |^2 |d\zeta| \right)^{1/2} \, \| F \|.
\]
The Szego kernel is itself in $H^{2}(\Omega)$ as a function of its left argument, hence
\[
K(z,z') = \frac{1}{L} \int_{\Gamma} K(\zeta, z') \overline{K(\zeta, z)} | dz |.
\]
By letting $z = z'$, we get $\frac{1}{L} \int_{\Gamma} | K(\zeta,z) |^2 |d\zeta| = K(z,z)$, hence
\[
| F(z) | \leq K(z,z)^{1/2},
\]
and
\[
|P(z)| = |F(z)| \, |\Phi(z)|^n \leq K(z,z)^{1/2} |\Phi(z)|^n.
\]
\end{proof}

An application of Proposition \ref{teo:Kzzeta} yields an upper bound for $P(z)$ as
\beq\label{eq:bound_P}
|P(z)|^2 \leq \frac{L}{\pi} |\Phi'(z)| \frac{|\Phi(z)|^2}{|\Phi(z)|^2 - 1} |\Phi(z)|^{2n}, \qquad z \in \Omega.
\eeq
It is a good match, up to a factor $\sqrt{2}$, with the absolute value of the asymptotic approximation (\ref{eq:gnz}) for $p_n(z)$ at $z = \infty$. (This $\sqrt{2}$ factor can again be traced to the fact that $\Gamma$ is a Jordan arc traversed once, not a Jordan curve. It is unclear to us that it can be removed in the context of non-asymptotic bounds.) 

However, near $\Gamma$ where $|\Phi(z)| = 1$, the bound is very loose. To formulate a better bound, consider the banana-shaped region bounded by $\Gamma_2 = \{ z : | \Phi(z) | = 2 \}$. 

\begin{lemma} Let $P(z)$ be a polynomial of degree $n$ such that $\| P \|^2 = \frac{1}{L} \int_\Gamma |P(z)|^2 |dz| = 1$. For all $z$ in the interior of $\Gamma_2$,
\beq\label{eq:bound_P_gamma2}
|P(z)|^2 \leq \frac{4 L}{\pi} \, \frac{1}{c \sqrt{1-c^2}} \, 2^{2n}.
\eeq
\end{lemma}
\begin{proof}
Since $P(z)$ is analytic, we can apply the maximum modulus principle inside $\Gamma_2$. In order to use the bound (\ref{eq:bound_P}) on $\Gamma_2$, we need an upper bound on $|\Phi'(z)|$. By passing to the $w = \Phi(z)$ variable via (\ref{eq:phiw}) and $\phi'(\Phi(z)) = 1/\Phi'(z)$, it is elementary but tedious to show that
\[
|\Phi'(z)| \leq \frac{1}{c \sqrt{1-c^2}} \frac{(|\Phi(z)| + c)^2}{|\Phi(z)|^2 - 1}.
\]
When $|\Phi(z)| = 2$, the bounds combine to give (\ref{eq:bound_P_gamma2}).
\end{proof}

We are now left with the task of bounding $\int_{|z|=1} |P(z)|^2 \, |dz|$ from equations (\ref{eq:bound_P}) and (\ref{eq:bound_P_gamma2}). Call $R_1$ the region defined by $|z| = 1$ and $|\Phi(z)| > 2$, while $R_2$ corresponds to $|z| = 1$ and $|\Phi(z)| \leq 2$
\begin{itemize}
\item For $R_1$, it is advantageous to pass to the $w = \Phi(z)$ variable via (\ref{eq:phiw}). The pre-image of the arc of $|z|=1$ limited by $|\Phi(z)| > 2$ is the arc of the circle $C$ of equation
\[
| w + \frac{1}{c} |^2 = \frac{1-c^2}{c^2},
\]
limited by $|w| > 2$. Using (\ref{eq:bound_P}), the two Jacobians cancel out and we get 
\[
\int_{R_1} |P(z)|^2 \, |dz| \leq \frac{L}{\pi} \int_{C} \frac{|w|^2}{|w|^2 - 1} |w|^{2n} |dw|.
\]
Parametrize $C$ using $w = - \frac{1}{c} - \frac{\sqrt{1-c^2}}{c} e^{i \theta}$ with $\theta \in [- \pi, \pi)$, so that the maximum of the integrand occurs when $\theta = 0$. The measure becomes $|dw| = \frac{\sqrt{1-c^2}}{c} d\theta$, and since $|w| > 2$, we have
\[
\int_{R_1} |P(z)|^2 \, |dz| \leq \frac{4L}{3\pi} \, c^{-2n-1} \, (1-c^2)^{1/2} \, \int_{-\pi}^{\pi} (2 - c^2 + 2 \sqrt{1-c^2} \cos \theta)^n \, d\theta.
\]
The integrand is handled using the bound $a + b \cos \theta \leq(a+b) e^{- \frac{\theta^2}{2(\frac{a}{b}+1)}}$ (valid for $\theta \in [-\pi, \pi)$ as long as $c < .85$). Let $\sigma = \frac{2-c^2}{2\sqrt{1-c^2}} + 1$, so that
\begin{align*}
\int_{R_1} |P(z)|^2 \, |dz| &\leq \frac{4L}{3 \pi} \, c^{-2n-1} \, (2(2-c^2))^{n} (1-c^2)^{1/2} \int_{-\pi}^{\pi} e^{- n \frac{\theta^2}{2 \sigma^2}} \, d\theta. \\
&\leq \frac{4L}{3 \sqrt{2 \pi n}} \, \left( \frac{4-2 c^2}{c^2} \right)^{n+\frac{1}{2}}.
\end{align*}
To get the last line, we have used the fact that $(1-c^2)^{1/2} \sigma \leq (2-c^2)^{1/2}$.

\item The contribution along $R_2$ is of a different order of magnitude. The endpoints $z_{\pm}$ of the corresponding arc of the unit circle can be obtained from $|\Phi(z)| = 2$ and $|z| = 1$, which reveals that $z_{\pm} = \phi(w_{\pm})$ with $w_{\pm} = 2 e^{\pm i \alpha}$ and $\cos \alpha = -5c/4$. Further elementary calculations (using $\mbox{acos}(x) \leq \frac{\pi}{2} \sqrt{1-x}$ for $x > 0$) show that the arc length between $z_{+}$ and $z_{-}$ is bounded by $2 \sqrt{2} \pi c$. Hence (\ref{eq:bound_P_gamma2}) implies that
\[
\int_{R_2} |P(z)|^2 \, |dz| \leq 8 \sqrt{2} L \, \frac{1}{\sqrt{1-c^2}} \, 2^{2n}.
\]

\end{itemize}

The upper bound on $\int_{|z|=1} |P(z)|^2 |dz|$ is then the sum of the contributions along $R_1$ and $R_2$. The former contribution always dominates the latter (up to multiplicative constants, either in the limit $c \to 0$ or $n \to \infty$), because our assumption that $y < \frac{1}{2}$ implies $c \leq \sqrt{2} / {2}$, and in turn $\left( \frac{4-2c^2}{c^2} \right)^n \geq 6^n \geq 4^n$. 
A bit of grooming results in the bound
\[
\lambda^{-1}_{\min}(G) \leq C \, \left( \frac{4-2c^2}{c^2} \right)^n, 
\]
where $C > 0$ is a reasonable numerical constant. An even shorter statement is $\lambda_{\min}(G) \geq C' \left( \frac{c}{4} \right)^{2n}.$

For the least singular value of $a_k(x)$, we get the lower bound
\beq\label{teo:lower-bound-least-sv}
C \, \left( \frac{c}{4} \right)^{n}, \qquad c = \sin(\pi y / 2).
\eeq

%
%
%
%
%
%

\section{Non-consecutive atoms}\label{sec:non-contiguous}

We now prove Lemma \ref{teo:eps2}.

Let $T$ for a set of $n+1$ integers that we denote $\tau_j$ (the fact that they are integers has no bearing on the forthcoming argument.) Let
\[
(A_T)_{\theta, j} = \frac{e^{i \tau_j \theta}}{\sqrt{2 \pi y}}, \qquad \theta \in [ -\pi y, \pi y ], \qquad \tau_j \in T.
\]
The Gram matrix $G_{j_1, j_2} = \int_{-\pi y}^{\pi y} (A_T)_{\theta, j_1} \overline{(A_T)_{\theta, j_2}} \, d\theta$ is invariant under translation of the $\tau_j$, hence so are its eigenvalues. We may therefore view $G$'s eigenvalues as functions of the differences $\tau_{j_1, j_2} = \tau_{j_2} - \tau_{j_1}$. Recall that $\lambda_{\min}(G) = \sigma^2_{\min}(A_T)$. 

\begin{definition}
We say that a function $f( \{ \tau_{j_1, j_2} \} )$ of some arguments $\tau_{j_1, j_2}$, for $0 \leq j_1 < j_2 \leq n$, is increasing if 
\[
f( \{ \tau'_{j_1, j_2} \} ) \geq f( \{ \tau_{j_1, j_2} \} ),
\]
provided $\tau'_{j_1, j_2} \geq \tau_{j_1, j_2}$ for all $0 \leq j_1 < j_2 \leq n$. Furthermore, $f$ is strictly increasing if
\[
f( \{ \tau'_{j_1, j_2} \} ) > f( \{ \tau_{j_1, j_2} \} ),
\]
provided at least one of the inequalities $\tau'_{j_1, j_2} \geq \tau_{j_1, j_2}$ is strict.
\end{definition}

\begin{theorem} (Monotonicity of $\lambda_{\min}(G)$ in $T$.) Fix $n+1$, the cardinality of $T$. There exists $y^* > 0$, such that for all $0 < y \leq y^*$, the eigenvalue $\lambda_{\min}(G)$ is an increasing function of the phase differences $\tau_{j_2} - \tau_{j_1}$. 
\end{theorem}

The theorem shows that, as long as $\tau_j$ are integers, the minimum eigenvalue of $G$ is minimized when $T$ is any set of $n+1$ consecutive integers. We conjecture that the result still holds when the restriction on $y$ is lifted.

\begin{proof}

It suffices to shows that $\lambda_{\min}(G)$ is \emph{strictly} increasing in the phase differences in the limit $y \to 0$, since the claim will also be also true for sufficiently small $y$ by continuity of $\lambda_{\min}(G)$ as a function of $y$.

Without loss, consider that $\theta \in [0, 2 \pi y]$ instead of $[ - \pi y, \pi y ]$. This transformation does not change the eigenvalues of $G$. Expand the complex exponential in Taylor series to get
\begin{align*}
v^* G v &= \sum_{j_1, j_2 = 0}^{n} \overline{v_{j_1}} v_{j_2} \; \frac{1}{2 \pi y} \int_{0}^{2 \pi y} e^{i \theta (\tau_{j_1} - \tau_{j_2})} d\theta \\ 
&= \sum_{m_1, m_2 \geq 0} \overline{q_{m_1}} q_{m_2} \frac{i^{m_1}}{m_1!} \frac{(-i)^{m_2}}{m_2!} \frac{1}{2 \pi y} \int_0^{2 \pi y} \theta^{m_1 + m_2} d \theta,
\end{align*}
where $q_{m}$ is the $m$-th moment of $v$ with respect to the $\tau_j$, i.e.,
\[
q_{m} = \sum_{j = 0}^{n} v_{j} \tau_{j}^m.
\]
One way to invert this relationship for $v_j$ is to write
\[
q = \bp q_{0:n} \\ q_{n+1:\infty} \ep = \bp M \\ N \ep v,
\]
with the square Vandermonde matrix
\[
M = \bp 1 & \cdots & 1 \\ \theta_0 & \cdots & \tau_{n} \\ \vdots & & \vdots \\ \theta^{n}_{0} & \cdots & \tau_{n}^{n} \ep,
\]
and then let $v = M^{-1} q_{0:n}$. The integral factor in the expression of $v^* G v$ is $\int_0^{2 \pi y} \theta^{m_1 + m_2} d\theta = (2 \pi y)^{m_1 + m_2 +1} H_{m_1, m_2}$, with $H_{m_1, m_2} = \frac{1}{m_1 + m_2 + 1}$ the Hilbert matrix. After further letting $D_y = \mbox{diag} \left( \frac{(2 \pi i)^{m}}{m!} y^{m} \right)$, we may express the Rayleigh quotient for $G$ in terms of $q$ as
\[
J = \frac{v^* G v}{v^* v} = y \; \frac{q^* D^*_y H D_y q}{q_{0:n}^* M^{-*} M^{-1} q_{0:n}}.
\]
Notice that the dependence on $y$ is only present in the diagonal factor $D_y$ (and the leading scalar factor $y$.) 
 
Since at least one of the first $n+1$ components of $q$ is nonzero, and $H$ is positive definite when restricted to $0 \leq m_1, m_2 \leq n$, the minimum of $J$  must be of exact order $y^{2n+1}$ as $y \to 0$. This is for instance the case when $q_{0:n} = (0, \ldots, 0, 1)^T$ and $q_{n+1:\infty} = 0$. More generally, values on the order of $y^{2n+1}$ can only be obtained when the weight of $q_{0:n}$ is predominantly on the last component.


In the limit $y \to 0$, we now observe that the contribution of $q_{n+1:\infty}$ is negligible in the numerator of $J$, since
\[
y \; (0 \;\; q^*_{n+1:\infty}) \, D^*_y H D_y  \bp 0 \\ q_{n+1:\infty} \ep = O(y^{2n+3}) \ll y^{2n+1}.
\] 
Denote by $H_n$ and $D_{y,n}$ the respective $0 \leq m_1, m_2 \leq n$ sections of $H$ and $D_y$. With $p$ a shorthand for $D_{y,n} \; q_{0:n}$, the problem has been reduced to proving that the (nonzero) limit as $y \to 0$ of 
\[
\min_{p \ne 0} \; y^{-2n} \, \frac{p^* H_k p}{p^* D_{y,n}^{-*} M^{-*} M^{-1}D^{-1}_{y,n} p}
\]
is strictly decreasing in the phase differences. We are in presence of the minimum generalized eigenvalue of the pencil $H_n - \mu B_y$, where
\[
B_y = y^{2n} D_{y,n}^{-*} M^{-*} M^{-1} D^{-1}_{y,n}.
\]
As $y > 0$, $B_y$ is invertible, hence all the generalized eigenvalues are positive. As $y \to 0$ however, $B_y$ degenerates to the rank-1 matrix
\[
B_0 = c_n \; mm^*,
\]
with $c_n = \frac{(2 \pi)^{2n}}{(n!)^2}$ and where $m^*$ is the $n$th (i.e., last) row of $M^{-1}$. In that case, all but one of the generalized eigenvalues become $+ \infty$. (Indeed, we can change basis to transform the pencil $H_n - \mu c_n mm^*$ into some $\tilde{H}_n - \mu e_1 e_1^T$, whose characteristic polynomial has degree 1.) It is convenient to call this generalized eigenvalue $\mu$; it depends in a continuous and differentiable manner on $y$ as $y \to 0$. In the limit $y \to 0$, the gradient of $\mu$ in the components of $m$ can be obtained by standard perturbation analysis as
\[
\nabla_m \mu = \frac{2 \mu}{m^* m} m^*.
\]
Interestingly, it does not depend on $H_n$, and only depends on $y$ through $\mu$.

The inverse of the Vandermonde matrix $M$ can be computed with Vieta's formulas, which yield a closed-form expression for $m$:
\[
m_j = (-1)^j \prod_{i \ne j} \frac{1}{\tau_i - \tau_j}, \qquad 0 \leq j \leq n.
\] 
(The other elements of $M^{-1}$ appear to be significantly more complicated.) Each component $|m_j|$ in absolute value is manifestly strictly decreasing in the phase differences. We wish to reach the same conclusion for the eigenvalue $\mu$. 

For any two vectors $m$ and $m'$ corresponding to different sets of phases $\tau_j$ and $\tau'_j$ such that
\[
\tau_j' - \tau_i' \geq \tau_j - \tau_i, \qquad  j > i,
\]
at least one of the inequalities being strict, it is clear that $|m_j'| \leq |m_j|$, with at least two of the inequalities being strict. It is also clear that $m$ and $m'$ can be connected in a continuous way so as to respect this monotonicity property, namely there exists a sequence $m(t)$ indexed by some parameter $t \in [0,1]$ such that 
\begin{itemize}
\item $m(0) = m$ and $m(1) = m'$;
\item $m(t)$ is piecewise differentiable with bounded derivative $\dot{m}(t)$;
\item the sign of $\dot{m}_j(t)$ matches that of $-m_j(t)$ componentwise; and 
\item $\dot{m}_j(t) \ne 0$ in at least two components at a time.
\end{itemize}
The corresponding values of $\mu = \mu(m(0))$ and $\mu' = \mu(m(1))$ obey
\begin{align*}
\mu' - \mu &= \int_0^1 \frac{d}{dt} \mu(m(t)) dt \\ 
&= \int_0^1 \dot{m}(t)^T \nabla_m \mu(m(t)) dt \\
&= 2 \int_0^1 \mu(m(t)) \frac{\dot{m}(t)^T m(t)}{m(t)^T m(t)} dt.
\end{align*}
By construction $\dot{m}(t)^T m(t) < 0$ for all $t$, hence we reach the desired conclusion that $\mu' < \mu$.

\end{proof}

%
%

\bibliography{bib2}
\bibliographystyle{plain}

\end{document}